\documentclass[12pt,a4paper,UTF8]{article}
\usepackage{amsmath,mathtools,amsfonts,amsmath,amssymb,mathrsfs,dsfont,cases,amsthm,bm,bbm}
\usepackage{color,xcolor}
\usepackage{titlesec}
\usepackage{geometry}
\usepackage{natbib}
\usepackage{graphicx,subfigure}
\usepackage[colorlinks,citecolor=blue,urlcolor=blue,linkcolor=blue]{hyperref}%
\usepackage{array,booktabs,makecell,tabularx,multirow} 

\newcommand{\cX}{\mathcal{X}}
\newcommand{\cY}{\mathcal{Y}}
\newcommand{\bbP}{\mathbb{P}}
\newcommand{\bbR}{\mathbb{R}}
\newcommand{\bbE}{\mathbb{E}}
\newcommand{\cN}{\mathcal{N}}
\newcommand{\diag}{\mathrm{diag}}
\newcommand{\Var}{\mathrm{Var}}

\newtheoremstyle{thry}
{1em}
{1em}
{\itshape\rmfamily}
{}
{\scshape\large}
{.}
{.5em}
{}
\theoremstyle{thry}

\newtheorem{Thm}{\indent Theorem}
\newtheorem{Lem}{\indent  Lemma}
\newtheorem{Ass}{\indent  Assumption}

\newtheorem{Cor}{\indent  Corollary}

\title{Two-Sample Test for Stochastic Block Models via the Largest Singular Value}
\author{Kang Fu, Jianwei Hu, Seydou Keita and Hang Liu \\ {\small \itshape School of Mathematics and Statistics, Central China Normal University, Wuhan}}
\date{\today}

\begin{document}
\maketitle
\begin{abstract}
The stochastic block model is widely used for detecting community structures in network data. However, the research interest of much literature focuses on the study of one sample of stochastic block models. How to detect the difference of the community structures is a less studied issue for stochastic block models. In this article, we propose a novel test statistic based on the largest singular value of a residual matrix obtained by subtracting the geometric mean of two estimated block mean effects from the sum of two observed adjacency matrices. We prove that the null distribution of the proposed test statistic converges in distribution to the Tracy-Widom distribution with index 1, and we show the difference between the two samples for stochastic block models can be tested via the proposed method. Further, we show that the proposed test has asymptotic power guarantee against alternative models. Both simulation studies and real-world data examples indicate that the proposed method works well.

\end{abstract}

\section{Introduction}\label{Introduction}

Network data analysis has become a popular research topic in many fields, including gene classification, social relationship investigation, and financial risk management. The study on network data with community structures has received increasing interest \cite{Newman:2004, Newman:2006, Steinhaeuser:2010}. In network data analysis, the stochastic block model proposed by \cite{Holland:1983} is a popular tool to fit network data with community structure, see, e.g., \cite{Snijders:1997, Nowicki:2001, Bickel:2009, Rohe:2011, Choi:2012, Jin:2015, Zhang:2016}. For a stochastic block model with $n$ nodes and $K$ blocks, all nodes can be clustered to $K$ disjoint block. Specifically, for a node $i$ belonging to block $u$, there is a mapping $g$ such that $g(i)=u$. In the stochastic block model, the observed network can be represented by an $n\times n$ adjacency matrix $A$, where the entries of $A$ represent whether an edge exists between two nodes. Given the block label $g$, the model assumes that each entry $A_{ij}$ of $A$ can be seen as a Bernoulli random variable with probability $P_{ij}=B_{g(i)g(j)}$, where $P$ and $B$ are $n\times n$ edge-probability matrix and $K\times K$ probability matrix, respectively. It is easy to see that given the pair $(g,B)$ up to community label permutations of nodes, a stochastic block model can be completely and uniquely determined. In this paper, we mainly focus on the undirect and unweight graph, and no self-loop exists, i.e., $A$ be a symmetric binary matrix and $A_{ii}=0$ for any $i$.

In the research on the stochastic block model, there are two fundamental problems: model selection and community detection. The model selection aims to estimate the number of blocks (or the number of communities) using the sample adjacency matrix $A$. To exactly estimate the number of communities, the majority of approaches have been proposed, see, for example, the recursive approach \cite{Zhao:2011}, the sequential testing methods \cite{Lei:2016,Hu:2021} and the likelihood-based methods \cite{Saldna:2017,Wang:2017,Hu:2020}. In addition, the purpose of community detection is to recover the community structure, that is, cluster all nodes into different communities so that the connections between nodes within the community are relatively dense and the nodes between the community are relatively sparse. For this issue, many efficient methods have also been proposed, including the modularity \citep{Newman:2006}, the profile-likelihood method \citep{Bickel:2009}, the variational method \citep{Daudin:2008}, the spectral clustering \citep{Rohe:2011,Jin:2015}, the pseudo-likelihood method \citep{Amini:2013}, and the profile-pseudo likelihood method \citep{Wang:2021}. Asymptotic properties of the estimation of community membership label have been established \citep{Rohe:2011, Zhao:2012, Choi:2012, Lei:2015, Sarkar:2015, Zhang:2016, Wang:2021}. To improve the flexibility of stochastic block models, Karrer and Newman \cite{Karrer:2011} proposed the degree-corrected stochastic block model, which introduces the additional degree-corrected parameters to fit the network with degree heterogeneity. Further, Airoldi et al. \cite{Airoldi:2008} proposed the mixed membership stochastic block model, in which a single node may belong to several different communities.

In the past decades, many works pay much attention to one sample in the study of stochastic block models. Usually, in social networks, an individual in one party may tend to develop a closer relationship with another party, see, e.g., \cite{Barnett:2016}. This situation may cause a problem, that is, the community structure of the network may change over time or other conditions. This also prompts us to consider whether the community structure of the two networks is the same. A direct approach is to judge whether two stochastic block models are identical by the hypothesis testing method. 

For the two-sample testing issue, Tang et al. \cite{Tang:2017} considered the two-sample test of the random dot product graph, and used the adjacency spectral embedding method to construct the test statistic. This testing procedure can judge whether two independent network samples are from an identical model. Further, Ghoshdastidar and von Luxburg \cite{Ghoshdastidar:2018} considered the two-sample testing issue of the general network model. They used the largest singular value of scaled and centralized matrix of the difference between two adjacency matrices to construct the test statistic and proved the asymptotic null distribution is a Tracy-Widom distribution. Similarly, Ghoshdastidar et al. \cite{Ghoshdastidar:2020} used the spectral norm and the Frobenius norm to construct two test statistics. However, this procedure needs to choose an appropriate threshold, but the threshold is hard to choose in practice. Recently, Chen et al. \cite{Chen:2021c} proposed a new testing statistic to test two general network samples, which simplified the statistic in \cite{Ghoshdastidar:2018}. Meanwhile, Chen et al. \cite{Chen:2021s} used the trace of a new construction matrix to get the test statistic and prove the null distribution convergence in distribution to the standard normal distribution. However, these methods require multiple samples, which implied that the procedure may not work well when we have less sample information (such as, we only have two network samples). In addition, these methods consider the testing problem $H_0:P_1=P_2$, where $P_1$ and $P_2$ are the edge-probability matrices of two network models, and require that the difference of edge-probability matrices is great. But in the stochastic block model, changes in the community structure may lead to small changes in the edge-probability matrix, which implies that the methods mentioned above may not work well. Hence, we need to consider an appropriate test statistic such that it can test the tiny difference when the network size is large. For the two-sample test of the stochastic block model, based on the locally smoothed adjacency matrix, Wu et al. \cite{Wu:2021} proposed a new statistic to test two stochastic block models. This method constructs the locally smoothed adjacency matrix by separating a community into serval non-overlapping neighboring sub-communities and averages the entries of the adjacency matrix in non-overlapping local neighborhoods within communities. However, this procedure can only be used when the number of communities is small and the process of construction is complex. Therefore, in this paper, we hope to propose a new test statistic that can be used in the case of $K$ diverging with $n$.

Specifically, let $\cX=(X_{ij})_{n\times n}$ and $\cY=(Y_{ij})_{n\times n}$ be two adjacency matrices with stochastic block structures $(g_x, B_x)$ and $(g_y, B_y)$, respectively. In this paper, we assume that the node labels (not community label $g$) of the two networks are predetermined and kept the same, that is, node $i$ in the network $\cX$ also appears in the network $\cY$. Hence we could regard $\cX$ and $\cY$ as repeated network measurements of a number of individuals. Now the test problem can be formulated as
\begin{equation}\label{eq:problem}
	H_0: (g_x, B_x) = (g_y, B_y)\quad \mathrm{v.s.}\quad H_1: (g_x, B_x) \neq (g_y, B_y).
\end{equation}
Intuitively, under the alternative hypothesis, there are four mutually exclusive cases: (i) $K_x \neq K_y$, (ii) $K_x = K_y, g_x = g_y$, but $B_x \neq B_y$, (iii) $K_x = K_y, B_x = B_y$, but $g_x \neq g_y$, and (iv) $K_x = K_y, g_x\neq g_y, B_x\neq B_y$. 

In this paper, we propose a two-sample test statistic to detect the difference between two stochastic block models. Under two observed adjacency matrices, the largest singular value of the residual matrix obtained by removing the geometric mean of two estimated block mean effects from the sum of two observed adjacency matrices is used to construct the test statistic. Intuitively, if $\cX$ and $\cY$ are generated from a common stochastic block model, and the entries $X_{ij}$ and $Y_{ij}$ of $\cX + \cY$, that is, the sum of two adjacency matrices, should be independent Binomial random variables. From this view, we consider combining two adjacency matrix samples to obtain a new combined adjacency matrix. Based on the largest singular value of the residual matrix of the combined adjacency matrix, we obtain a new test statistic. We show that the asymptotic null distribution of the test statistic is the Tracy-Widom distribution with index 1 when $K_x = K_y = o(n^{1/6})$.

Our testing procedure proposed in this paper goes as follows. First, we estimate the numbers of communities $K_x$ and $K_y$ using the existing estimation methods. Then, based on the estimators $\widehat{K}_x$ and $\widehat{K}_y$, we can get the consistent estimators $\widehat{g}_x$ and $\widehat{g}_y$. Second, we estimate entries of $B_x$ and $B_y$ by the sample proportions based on $(\cX, \widehat{g}_x)$ and $(\cY, \widehat{g}_y)$. Third, we combine the two samples $\cX + \cY$, and use $\sqrt{\widehat{P}^x_{ij} \cdot \widehat{P}^y_{ij}}$ to center and rescale $X_{ij} + Y_{ij}$, and obtain a residual matrix. Finally, we obtain the test statistic by computing the largest singular value of the residual matrix. The detailed mathematical expression and the theoretical results of the test statistic will be given in Section \ref{method}. The basic view of this method is that the entries of the sum of two adjacency matrices are Binomial random variables when the null hypothesis is true. Hence, the residual matrix will approximate a generalized Wigner matrix: a symmetric random matrix with independent mean zero upper diagonal entries. According to the random matrix theory, the largest eigenvalue follows a Tracy-Widom distribution. On the other hand, under the alternative hypothesis, the entries of the sum of two adjacency matrices are not Binomial random variables, and the entries of the residual matrix will not approximate a generalized Wigner matrix. This shows that the test statistic can successfully separate $H_0$ and $H_1$ in \eqref{eq:problem}.

Now, we give some notations. For a square matrix $A$, $\diag(A)$ denotes the diagonal matrix induced by $A$. For any $n \times n$ symmetric matrix $A$, $\lambda_j(A)$ denotes its $j$th largest eigenvalue value, ordered as $\lambda_1(A) \geq \lambda_2(A) \geq \cdots \geq \lambda_n(A)$, and $\sigma_1(A)$ is the largest singular value. For a matrix $A$, the notation $\|A\|$ denotes the spectral norm.

The remainder of the article is organized as follows. In Section \ref{method}, we introduce the new test statistic, and state its asymptotic null distribution and asymptotic power. Simulation studies and real-world data examples are given in Sections \ref{Simulation} and \ref{Real}, respectively. Some discussions are given in Section \ref{Discussion}. All technical proofs are postponed to the Appendix.

\section{A New Two-Sample Test for Stochastic Block Models}\label{method}

Consider a stochastic block model on $n$ nodes. Given a number of communities $K$ and a community label $g$, the maximum likelihood estimator of $B$ is given by
\begin{equation}\label{eq:estB}
	\widehat{B}_{uv} = \begin{dcases}
 \dfrac{\sum_{i\in\cN_{u}, j\in\cN_{v}}X_{ij}}{n_u n_v}, \quad & u \neq v,\\
 \dfrac{\sum_{i,j\in\cN_{u},i\neq j}X_{ij}}{n_u (n_v-1)}, \quad & u = v,	
 \end{dcases}
\end{equation}
where $\cN_{u} = \left\{i: g(i) = u\right\}$ for $1 \leq u \leq K$, $i$ stands for the node label and $n_u = \left|\cN_{u}\right|$. Throughout, we denote $\cN_{u}^x = \left\{i: g_x(i) = u\right\}$ and $n_u^x = \left|\cN_{u}^x\right|$, where $x$ can be replaced by $y$.

For a single observed adjacency matrix $\cX$, Lei \cite{Lei:2016} considered using a test statistic, based on the largest singular value, to test the following hypothesis test:
\[
H_0:K = K_0\quad \mathrm{v.s.}\quad H_1: K>K_0,
\]
where, we use $K$ to denote the true number of communities, and use $K_0$ to denote a hypothetical number of communities. Let the centered and rescaled adjacency matrix $\tilde{X}$ be
\[
\tilde{X} = \begin{dcases}
	\dfrac{X_{ij} - \widehat{P}_{ij}}{\sqrt{(n-1)\widehat{P}_{ij}(1-\widehat{P}_{ij})}}, & i\neq j, \\
	0, & i=j,
\end{dcases}
\]
where $\widehat{P}_{ij} = \widehat{B}_{g_0(i)g_0(j)}$, as defined in \eqref{eq:estB}. Under the null hypothesis $H_0 : K = K_0$, if $K = o(n^{1/6})$, Lei \cite{Lei:2016} showed that
\[
n^{2/3}[\lambda_1(\tilde{X}) - 2] \stackrel{d}{\longrightarrow} TW_1\ \text{and}\ n^{2/3}[-\lambda_n(\tilde{X}) - 2] \stackrel{d}{\longrightarrow} TW_1,
\]
where $TW_1$ denotes the Tracy-Widom distribution with index 1 and ``$\stackrel{d}{\longrightarrow}$" denotes convergence in distribution.. Further, for test $H_0 : K=K_0$, Lei \cite{Lei:2016} proposed to obtain $\widehat{g}$ using consistent clustering methods (under $K=K_0$) and considered the following test statistic:
\[
T_{n,K_0}=\max\{n^{2/3}[\lambda_1(\tilde{X}) - 2], n^{2/3}[-\lambda_n(\tilde{X}) - 2]\},
\]
where $g_0$ in $\tilde{X}$ has been replaced by $\widehat{g}$. Given the test level $\alpha$, the corresponding rejection rule is then
\[
\mathrm{Reject:}\ H_0: K = K_0\ \mathrm{if}\ T_n\geq t_{(\alpha/2)},
\]
where $t_{\alpha}$ is the $\alpha$th quantile of the $TW_1$ distribution for $\alpha\in(0,1)$. 

In this paper, we aim to develop a new test statistic that allows this statistic to test two samples under the case of $K$ diverging with $n$. For two samples $\cX$ and $\cY$ from stochastic block models $(g_x, B_x)$ and $(g_y, B_y)$, there are many methods for estimating $\widehat{K}_x$ and $\widehat{K}_y$. In particular, in this paper, we may use the corrected Bayesian information criterion in \cite{Hu:2020} to estimate $K_x$ and $K_y$. Given $\widehat{K}_x$ and $\widehat{K}_y$, the community labels $g_x$ and $g_y$ can be estimated, denote by $\widehat{g}_x$ and $\widehat{g}_y$, by some existing consistent community detection procedures. Then, we can get the $\widehat{B}^x$ and $\widehat{B}^y$ by equation \eqref{eq:estB}.


To derive a two-sample test for stochastic block models, we extend the method in \cite{Lei:2016} to two samples of stochastic block models. For binomial random variables, it is easy to know that the sum of binomial random variables also follows the binomial distribution, i.e.,
\[
\xi+\eta \sim B(n_1+n_2, p)\quad \mathrm{for}\quad \xi \sim B(n_1, p)\ \mathrm{and}\ \eta\sim B(n_2, p).
\]

In the stochastic block model, the entries $X_{ij}$ (or $Y_{ij}$) of the unweight adjacency matrix $\cX$ (or $\cY$) can be seen as Bernoulli random variables with parameter $P^x_{ij} = B^x_{g_x(i)g_x(j)}$ (or $P^y_{ij} = B^y_{g_y(i)g_y(j)}$), so that
\[
\bbE\{\cX\} = P^x - \diag(P^x)\ \mathrm{and}\ \bbE\{\cY\} = P^y - \diag(P^y).
\]
Let $\tilde{A}^*$ be
\begin{equation}\label{eq:As}
\tilde{A}^* = \begin{dcases}
	 \dfrac{X_{ij} + Y_{ij} - 2P_{ij}}{\sqrt{(n-1)\cdot2P_{ij}(1-P_{ij})}}, & i\neq j \\
	 0, & i=j,
	 \end{dcases}
\end{equation}
where $P_{ij} = \sqrt{P_{ij}^x \cdot P_{ij}^y}$. Here, the reason why we use geometric mean instead of other mean is that when the alternative hypothesis is true, the singular value of the residual matrix will become larger, which helps us get better power of the test. Notice that, under $H_0$, $X_{ij} + Y_{ij}$ follows a binomial distribution with parameters $n = 2$ and $P_{ij} = P^x_{ij} = P^y_{ij}$, i.e., $X_{ij} + Y_{ij} \sim B(2, P_{ij})$. Hence $\tilde{A}^*$ is a general Wigner matrix, satisfying $\bbE\{\tilde{A}^*_{ij}\} = 0$ for all $(i, j)$ and $\sum_j\Var\{\tilde{A}^*_{ij}\} = 1$ for all $i$. In the random matrix theory, the asymptotic distribution of the extreme eigenvalues of the general Winger matrix has been wildly studied. In particular, according to results in \cite{Erdos:2012} and \cite{Lee:2014} we have
\begin{equation}\label{eq:eigcon}
	n^{2/3}[\lambda_1(\tilde{A}^*)-2]\stackrel{d}{\longrightarrow} TW_1\ \mathrm{and}\ n^{2/3}[-\lambda_n(\tilde{A}^*)-2]\stackrel{d}{\longrightarrow} TW_1,
\end{equation}
\noindent
We formally state and prove this result as Lemma \ref{Lemma:1} in Appendix.

The matrix $\tilde{A}^*$ involves unknown model parameters and cannot be used as a test statistic directly. Hence, a natural method is replacing these parameters with their estimators. we first obtain consistent estimators $\widehat{K}_x$, $\widehat{K}_y$, $\widehat{g}_x$ and $\widehat{g}_y$ by the method described above, then estimate $\widehat{B}_x$ and $\widehat{B}_y$ using equation \eqref{eq:estB}. Based on the estimators $(\widehat{g}_x, \widehat{B}_x)$ and $(\widehat{g}_y, \widehat{B}_y)$, the empirically centered and rescaled adjacency matrix $\tilde{A}$ can be obtained:

\begin{equation}\label{eq:At}
	\tilde{A} = \begin{dcases}
 \dfrac{X_{ij} + Y_{ij} - 2\widehat{P}_{ij}}{\sqrt{(n-1)\cdot 2\widehat{P}_{ij}(1-\widehat{P}_{ij})}}, & i\neq j, \\
 0, & i = j,
 \end{dcases}
\end{equation}
where 
\[
\widehat{P}_{ij} = \sqrt{\widehat{P}_{ij}^x \cdot \widehat{P}_{ij}^y} = \sqrt{\widehat{B}^x_{\widehat{g}_{x}(i)\widehat{g}_{x}(j)} \cdot \widehat{B}^y_{\widehat{g}_{y}(i)\widehat{g}_{y}(j)}}.
\]
It is natural to conjecture that under the null hypothesis $(g_x, B_x) = (g_y, B_y)$ and when the estimates $(\widehat{g}_x, \widehat{B}_x)$ and $(\widehat{g}_y, \widehat{B}_y)$ are accurate enough, after $\tilde{A}^*$ replaced by $\tilde{A}$ in \eqref{eq:eigcon}, the convergence in \eqref{eq:eigcon} still holds. Therefore, we can use the largest singular value of $\tilde{A}$ as our test statistic:
\[
	T_n = n^{2/3}[\sigma_1(\tilde{A})-2].
\]
Hence, the corresponding level $\alpha$ rejection rule is
\begin{equation}\label{eq:rej}
	\mathrm{Reject}\ \ H_0, \quad T_n\geq t_{(\alpha/2)},
\end{equation}
where $t_{(\alpha/2)}$ is the $\alpha/2$ upper quantile of the $TW_1$ distribution for $\alpha\in(0, 1)$. Since
\[
\sigma_1(\tilde{A}) = \max\{\lambda_1(\tilde{A}), -\lambda_n(\tilde{A})\},
\]
the test statistic also can be written as:
\[
T_n = \max\{n^{2/3}[\lambda_1(\tilde{A})-2], n^{2/3}[-\lambda_n(\tilde{A})-2]\}.
\]

In Section \ref{null} below, we formally state and prove the validity of our test statistic $T_n$ in Theorem \ref{Thm:null} by establishing the asymptotic null distributions of both the largest and smallest eigenvalues $\tilde{A}$ under $H_0$.

\subsection{The Asymptotic Null Distributions}\label{null}

The asymptotic distribution of the test statistic $T_n$ under the null hypothesis depends on the accuracy of the estimated community labels $\widehat{g}_x$ and $\widehat{g}_y$ and the estimated probability matrices $\widehat{B}_x$ and $\widehat{B}_y$. In order to consider the asymptotic behavior of $T_n$, we first consider the following two assumptions:

\begin{Ass}\label{Ass:1}
	The entries of $B_x$ and $B_y$ are uniformly bounded away from 0 and 1, and $B_x$ and $B_y$ have no identical rows.
\end{Ass}

\begin{Ass}\label{Ass:2}
	There exists constants $c_x$ and $c_y$ such that 
	\[
	\min_{1\leq k \leq K_x} n_k^x \geq c_xn/K_x\quad \mathrm{and}\quad \min_{1\leq k \leq K_y} n_k^y \geq c_yn/K_y.
	\]
\end{Ass}

Assumption \ref{Ass:1} requires that the entries of two probabilities $B_x$ and $B_y$ are bounded away from 0 and 1. At the same time, Assumption \ref{Ass:1} requires that $B_x$ and $B_y$ are identifiable. These assumptions are the basic conditions in the study of stochastic block models \cite{Lei:2016,Wang:2017,Hu:2021}. Assumption \ref{Ass:2} requires each community for $\cX$ and $\cY$ has size at least proportional to $n/K_x$ and $n/K_y$, respectively. This a mild condition. For example, it can be guaranteed almost surely when the community label $g$ is generated from a multinomial distribution with $n$ trials and probability $\bm{\pi} = (\pi_1, \cdots, \pi_K)$ such that $\min_{u}\pi_u\geq C_1/K$.

We now give the asymptotic results of both the largest and the smallest eigenvalues of $\tilde{A}$ and delay the proof to the Appendix.

\begin{Thm}\label{Thm:null}
	Suppose that Assumptions (\ref{Ass:1}) and (\ref{Ass:2}) hold. Let $\tilde{A}$ be given as in \eqref{eq:At} using corresponding plug-in estimator $(\widehat{g}_x, \widehat{B}_x)$ and $(\widehat{g}_y, \widehat{B}_y)$. Then under the null hypothesis $H_0: (g_x, B_x) = (g_y, B_y)$, as $n\rightarrow\infty$, if $K_x,K_y = O(n^{1/6-\tau})$ for some $\tau > 0$, we have
	\[
	n^{2/3}[\lambda_1(\tilde{A})-2]\stackrel{d}{\longrightarrow} TW_1\ \mathrm{and}\ n^{2/3}[-\lambda_n(\tilde{A})-2]\stackrel{d}{\longrightarrow} TW_1
	\]
\end{Thm}

A natural consequence of Theorem \ref{Thm:null} is that we can bound the asymptotic type I error for the rejection rule \eqref{eq:rej}:
\begin{align*}
	&\ \bbP\{T_n \geq t(\alpha/2)\} \\
	\leq &\ \bbP\{n^{2/3}[\lambda_1(\tilde{A})-2]\geq t(\alpha/2)\} + \bbP\{n^{2/3}[-\lambda_n(\tilde{A})-2]\geq t(\alpha/2)\}\\
	= &\ \alpha/2 + o(1) + \alpha/2 + o(1) = \alpha + o(1).
\end{align*}
which implies the reject rule is reasonable. Therefore, we have the following corollary.
\begin{Cor}\label{Cor:size}
	Under the assumptions of Theorem \ref{Thm:null}, given the level $\alpha$, we have
	\[
	\bbP\{T_n \geq t(\alpha/2)\}\rightarrow\alpha.
	\]
\end{Cor}

\subsection{The Asymptotic Power}
In this subsection, we study the asymptotic power of the proposed test. The following theorem gives a lower bound of the growth rate of test statistic $T_n$ under the alternative hypothesis.

\begin{Thm}\label{Thm:power}
	Suppose that Assumptions (\ref{Ass:1}) and (\ref{Ass:2}) hold. Let $\tilde{A}$ be given as in \eqref{eq:At} using corresponding plug-in estimator $(\widehat{g}_x, \widehat{B}_x)$ and $(\widehat{g}_y, \widehat{B}_y)$. Let $\delta_x$ and $\delta_y$ be the smallest $\infty$ distance among all pairs of distinct rows of $B_x$ and $B_y$, respectively. Then under the alternative hypothesis $H_1: (g_x, B_x) \neq (g_y, B_y)$, we have
	\[
	\sigma_1(\tilde{A})\geq\dfrac{1}{2\sqrt{2}}[(\delta_x\wedge\delta_y)(c_x\wedge c_y)nK^{-2}] + O_p(1).
	\]
\end{Thm}

Theorem \ref{Thm:power} is proved in Appendix. The main view is that if the two samples come from different stochastic block models, the sum of corresponding entries of two adjacency matrices cannot be seen as Binomial distribution, and hence it is impossible to obtain a reasonable residual matrix. Meanwhile, the above Theorem \ref{Thm:power} shows that $\sigma_1(\tilde{A})$ diverges at least rate $nK^{-2}$ as $n\rightarrow\infty$ under alternative hypothesis $H_1$. Then, we have the following corollary saying that the asymptotic power is almost equal to 1.

\begin{Cor}\label{Cor:power}
	Under the assumptions of Theorem \ref{Thm:power}, we have
	\[
	\bbP_{H_1}\{T_n > t_{\alpha/2}\}\rightarrow1.
	\]
\end{Cor}

Therefore, the Corollarys \ref{Cor:size} and \ref{Cor:power} suggest that the null and the alternative hypothesis are well separated, and our proposed test is asymptotically powerful against alternative hypothesis $H_1: (g_x, B_x) \neq (g_y, B_y)$.

\section{Simulation}\label{Simulation}

In this section, we implement numerical simulations to evaluate the performance of our proposed method. In the stochastic block model setting, the correction Bayesian information criterion and the spectral clustering method are used to estimate the number of communities and the community label, respectively. In the simulation, we consider the test statistic $T_n  = \max\{n^{2/3}[\lambda_1(\tilde{A})-2], n^{2/3}[-\lambda_n(\tilde{A})-2]\}$. Further, we compare our method with the TST-MD in \cite{Chen:2021c} and TST-S in \cite{Chen:2021s}.

\subsection{Simulation 1: The Null Distribution Under Stochastic Block Models and a Bootstrap Correction}

In the first simulation, we verify the result in Theorem \ref{Thm:null} by examining the finite sample null distribution of the test statistic $T_n$. Similarly, since the speed of convergence to a limiting distribution may be slow, we consider the bootstrap correction of the finite sample. The bootstrap correction method was first proposed in \cite{Bickel:2015} and later considered in \cite{Lei:2016} and \cite{Hu:2021}. Here, we extend their methods to our setting.

For two adjacency matrices $\cX$ and $\cY$ and null hypothesis $H_0: (g_x, B_x) = (g_y, B_y)$, the bootstrap corrected test statistic is calculated as the following:

1. Using consistent estimated methods to estimate $\widehat{K}_x = \widehat{K}_y = \widehat{K}$, then get the estimators $(\widehat{g}_x,\widehat{B}_x)$ and $(\widehat{g}_y,\widehat{B}_y)$;

2. Calculate $\tilde{A}$ as in \eqref{eq:At} and its the largest and the smallest eigenvalues $\lambda_1(\tilde{A}), \lambda_n(\tilde{A})$, respectively.

3. For $m = 1,2,\cdots,M$, 

\qquad (i) Let $\cX^{(m)}$ and $\cY^{(m)}$ be two adjacency matrices independently generated from probability matrix $\widehat{P}_{ij} = \sqrt{\widehat{P}_{ij}^x \cdot \widehat{P}_{ij}^y} = \sqrt{\widehat{B}_{\widehat{g}_x(i)\widehat{g}_x(j)}^x \cdot \widehat{B}_{\widehat{g}_y(i)\widehat{g}_y(j)}^y}$.

\qquad (ii) Let $\tilde{A}^{(m)} = (\tilde{A}^{(m)}_{ij})_{1\leq i,j \leq n}$, where $\tilde{A}^{(m)}_{ii} = 0$ and
\[
\quad \tilde{A}^{(m)}_{ij} = \dfrac{X_{ij}^{(m)} + Y_{ij}^{(m)} - 2\widehat{P}_{ij}}{\sqrt{(n-1)\cdot 2\widehat{P}_{ij}(1-\widehat{P}_{ij})}}, \quad 1\leq i < j\leq n.
\]

\qquad (iii) Let $\lambda_1^{(m)}$ and $\lambda_n^{(m)}$ be the largest and the smallest eigenvalues of $\tilde{A}^{(m)}$, respectively.

4. Let $(\widehat{\mu}_1, \widehat{\sigma}_1)$ and $(\widehat{\mu}_n, \widehat{\sigma}_n)$ be the sample mean and variance of $(\lambda_1^{(m)}: 1\leq m\leq M)$ and $(\lambda_n^{(m)}: 1\leq m\leq M)$, respectively.

5. The bootstrap corrected test statistic is calculated as
\[
T_{n}^{boot} = \mu_{\mathrm{tw}} + \sigma_{\mathrm{tw}}\max\left\{\dfrac{\lambda_1(\tilde{A}) - \widehat{\mu}_1}{\widehat{\sigma}_1}, \dfrac{-\lambda_n(\tilde{A}) - \widehat{\mu}_n}{\widehat{\sigma}_n}\right\},
\]
where $\mu_{\mathrm{tw}}$ and $\sigma_{\mathrm{tw}}$ are the mean and standard deviation of the Tracy-Widom distribution, respectively. In our simulations, we choose $M = 50$.

In Figure \ref{figure:TW}, we plot the distribution of the test statistic $T_n$ and the bootstrap corrected test statistic $T_n^{boot}$ based on 1000 data replications. In this simulation, we set $K_x = K_y =3$ and use equal-sized communities. The connection probability between communities $u$ and $v$ is $B_{uv}^x = B_{uv}^y = 0.1(1 + 4\times\mathbf{1}(u=v))$. Let sample sizes be $n = 600$ and $n = 1200$. We use the consistent method to get $\widehat{g}_x$ and $\widehat{g}_y$, then get $\widehat{B}_x$ and $\widehat{B}_y$. It is clear that the finite sample null distribution is systematically different from the limiting distribution when $n = 600$, and the difference is reduced but still visible when $n = 1200$. It is worth noting that although the limit distribution of finite sample shifts to the right relative to the limit distribution, there is little difference between the right tail of the distribution of finite sample and that of the limit distribution.

\begin{figure*}[h]
	\includegraphics[width = 8cm]{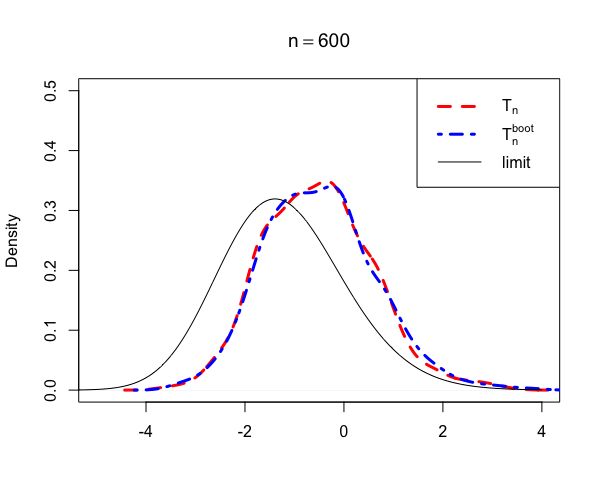}
	\includegraphics[width = 8cm]{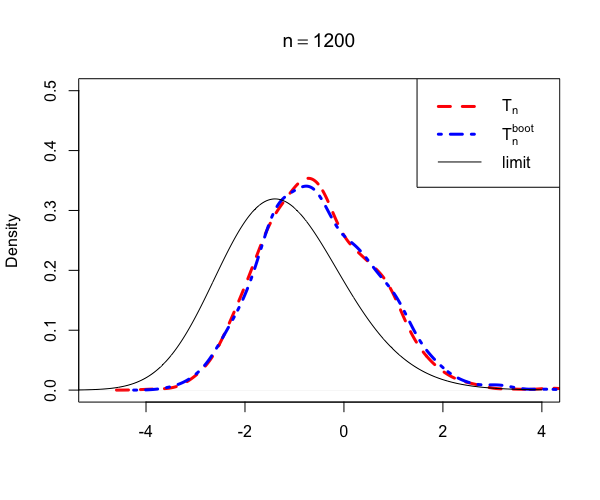}
	\caption{Null densities under the stochastic block model in Simulation 1 with $n = 600$ (left plot) and $n = 1200$ (right plot). The red dashed lines, blue dash-dotted lines, and black solid lines show the densities of the test statistic $T_n$, the bootstrap corrected test statistic $T_n^{boot}$, and the theoretical limit, respectively.}
	\label{figure:TW}
\end{figure*}

\subsection{Simulation 2: Empirical Size for Hypothesis \eqref{eq:problem}}

In this subsection, we investigate the empirical sizes of the proposed testing procedure for varying $K_x$, $K_y$, $B_x$, and $B_y$ under the null hypothesis. We set the connection probability between communities $u$ and $v$ as $B_{uv}^x = B_{uv}^y = r (1 + 3\times\mathbf{1}(u=v))$, where $r$ be used to measure the sparsity of the network. Let $K_x, K_y \in \{2,3,4,5,6,10\}, r\in\{0.05, 0.1, 0.2\}$, and the size of each community be 200. Table \ref{Table:size} reports the empirical size based on 200 data replications. From Table \ref{Table:size}, the empirical sizes of $T_n$ and $T_n^{boot}$ are closer to the nominal level and there is little difference between $T_n$'s Type I errors and $T_n^{boot}$'s Type I errors. When $T_n$ has a high empirical size, the empirical size can be reduced by bootstrap corrected statistic $T_n^{boot}$. Compared with the sparse situation ($r=0.05$), the test statistic works better. Meanwhile, we can see that the TST-MD and TST-S are not robust. In the majority of cases, the empirical sizes of TST-MD and TST-S tend to letter undersized.

\begin{table*}[h]
\setlength{\belowcaptionskip}{0.5cm}
\centering
\caption{Empirical size at nominal level $\alpha = 0.05$ for hypothesis test $H_0: (g_x, B_x) = (g_y, B_y)\ \mathrm{v.s.}\ H_1: (g_x, B_x) \neq (g_y, B_y)$. In the right half of the table, the values in the parentheses are the empirical size of TST-MD (left) and TST-S (right), respectively.}
\label{Table:size}
\begin{tabular*}{\textwidth}{@{\extracolsep{\fill}}cccccccc}
\hline
\multirow{2}{*}{} & \multicolumn{3}{c}{$T_n$} & & \multicolumn{3}{c}{$T_n^{boot}$}                                   \\ \cline{2-4} \cline{6-8} 
& $r=0.05$ & $r=0.1$ & $r=0.2$ & & $r=0.05$ & $r=0.1$ & $r=0.2$ \\ \hline
$K=2$ & 0.08 &	0.04 &	0.01 & & 0.06 (0.075, 0.075) & 0.05 (0.04, 0.03) & 0.04 (0.02, 0.06) \\
$K=3$ & 0.10 & 0.05 & 0.05 & & 0.05 (0.055, 0.065) & 0.05 (0.025, 0.06) & 0.09(0.015, 0.055) \\
$K=4$ & 0.05 & 0.04 & 0.04 & & 0.04 (0.045, 0.055) & 0.05 (0.015, 0.035) & 0.05 (0, 0.015) \\
$K=5$ & 0.07 & 0.04 & 0.04 & & 0.02 (0.045, 0.03) & 0.04 (0.01, 0.035) & 0.05 (0.01,0.03) \\
$K=6$ & 0.05 & 0.04 & 0.05 & & 0.04 (0.03, 0.035) & 0.04 (0.02, 0.04) & 0.06 (0, 0.055) \\
$K=10$ & 0.08 & 0.05 & 0.03 & & 0.04 (0.005, 0.045) & 0.06 (0, 0.005) & 0.05 (0.01, 0.025) \\
\hline
\end{tabular*}
\end{table*}

\subsection{Simulation 3: Empirical Power for Hypothesis \eqref{eq:problem}}

In this subsection, we compare the power for the test statistic $T_n$, the bootstrap corrected statistic $T_n^{boot}$, TST-MD, and TST-S against various alternative cases under hypothesis test problem \eqref{eq:problem}. We mainly consider three difference alternative cases: i) $g_x = g_y$ but $B_x\neq B_y$; ii) $B_x = B_y$ but $g_x \neq g_y$, iii) $K_x\neq K_y$. For the first alternative, we generate $\cX$ from the edge probability between communities $u$ and $v$ as $B_{uv}^x = r (0.5 + 3\times\mathbf{1}(u=v))$, and generate $\cY$ from the connection probability between communities $u$ and $v$ as $B_{uv}^y = r (3 + 5\times\mathbf{1}(u=v))$ for $r \in \{0.01, 0.05, 0.1\}$. We set the size of each block to be 200. For the second alternative, we set the connection probability between communities $u$ and $v$ as $B_{uv}^x = B_{uv}^y = r (1 + 3\times\mathbf{1}(u=v))$ for $r \in \{0.01, 0.05, 0.1\}$. For the community label of $\cX$, we set $g_x = (1, \cdots, 1, \cdots, K_x, \cdots, K_x)$. For the community membership label of $\cY$, the community membership label of $g_y$ is generated from a multinomial distribution with $n$ trials and probability $\bm{\pi} = (1/K_y, \cdots, 1/K_y)$. For the two settings, we set $K_x, K_y \in\{2,3,4,5,6,10\}$. For the third alternative, we let $K_x\in\{1, 2,3,4,6,10\}$ and $K_y=K_x+2$. We generate $cX$ and $\cY$ with edge probability $4r$ within community and $r$ between communities for $r\in\{0.05,0.1,0.2\}$. The sample size is $n=1200$ with equal-sized communities. For both alternatives, Tables \ref{table:power1}-\ref{table:power3} report the empirical power for hypothesis test \eqref{eq:problem} based on 200 data replication. It can be found that the proposed method and TST-MD can successfully detect all three types of alternative hypotheses, which indicates that our proposed two-sample test is powerful. In the second alternative model, however, the empirical power of TST-S is less than 1. As discussed in Section 1. the TST-S method will not separate the null hypothesis and alternative hypothesis well when the difference between the two models is small. For the third alternative model, we consider the case of $K_x\neq K_y$. In fact, when $\widehat{K}_x\neq\widehat{K}_y$, it is natural that $\widehat{g}_x\neq\widehat{g}_y$ and $\widehat{B}_x\neq\widehat{B}_y$. Then, since $\widehat{K}_x$ and $\widehat{K}_y$ are the consistent estimator of $\widehat{K}_x$ and $\widehat{K}_y$. we can reject the null hypothesis with probability tending to 1.

\begin{table*}[h]
\setlength{\belowcaptionskip}{0.5cm}
\centering
\caption{Empirical power of $T_n$ and $T_n^{boot}$ at nominal level $\alpha = 0.05$ for the first alternative. In the right half of the table, the values in the parentheses are the power of TST-MD (left) and TST-S (right), respectively.}
\label{table:power1}
\begin{tabular*}{\textwidth}{@{\extracolsep{\fill}}cccccccc}
\hline
\multirow{2}{*}{} & \multicolumn{3}{c}{$T_n$}   & \multicolumn{3}{c}{$T_n^{boot}$}  \\ \cline{2-4} \cline{6-8} 
                  & $r=0.01$ & $r=0.05$ & $r=0.1$ & & $r=0.01$ & $r=0.05$ & $r=0.1$ \\ \hline
$K=2$ & 1 & 1 & 1 & & 1 (1, 1) & 1 (1, 1) & 1 (1, 1) \\
$K=3$ & 1 & 1 & 1 & & 1 (1, 1) & 1 (1, 1) & 1 (1, 1) \\
$K=4$ & 1 & 1 & 1 & & 1 (1, 1) & 1 (1, 1) & 1 (1, 1) \\
$K=5$ & 1 & 1 & 1 & & 1 (1, 1) & 1 (1, 1) & 1 (1, 1) \\
$K=6$ & 1 & 1 & 1 & & 1 (1, 1) & 1 (1, 1) & 1 (1, 1) \\
$K=10$ & 1 & 1 & 1 & & 1 (1, 1) & 1 (1, 1) & 1 (1, 1) \\ \hline
\end{tabular*}
\end{table*}

\begin{table*}[h]
\setlength{\belowcaptionskip}{0.5cm}
\centering
\caption{Empirical power of $T_n$ and $T_n^{boot}$ at nominal level $\alpha = 0.05$ for the second alternative. In the right half of the table, the values in the parentheses are the power of TST-MD (left) and TST-S (right), respectively.}
\label{table:power2}
\begin{tabular*}{\textwidth}{@{\extracolsep{\fill}}cccccccc}
\hline
\multirow{2}{*}{} & \multicolumn{3}{c}{$T_n$}   & \multicolumn{3}{c}{$T_n^{boot}$}  \\ \cline{2-4} \cline{6-8} 
                  & $r=0.01$ & $r=0.05$ & $r=0.1$ & & $r=0.01$ & $r=0.05$ & $r=0.1$ \\ \hline
$K=2$ & 1 & 1 & 1 & & 1 (1, 0.055) & 1 (1, 0.015) & 1 (1, 0.335) \\
$K=3$ & 1 & 1 & 1 & & 1 (1, 0.09) & 1 (1, 0.2) & 1 (1, 0.505) \\
$K=4$ & 1 & 1 & 1 & & 1 (1, 0.03) & 1 (1, 0.22) & 1 (1, 0.565) \\
$K=5$ & 1 & 1 & 1 & & 1 (1, 0.0.085) & 1 (1, 0.0.225) & 1 (1, 0.615) \\
$K=6$ & 1 & 1 & 1 & & 1 (0.985, 0.06) & 1 (1, 0.215) & 1 (1, 0.53) \\
$K=10$ & 1 & 1 & 1 & & 1 (0.995, 0.045) & 1 (1, 0.135) & 1 (1, 0.505) \\ \hline
\end{tabular*}
\end{table*}

\begin{table*}[h]
\setlength{\belowcaptionskip}{0.5cm}
\centering
\caption{Empirical power of $T_n$ and $T_n^{boot}$ at nominal level $\alpha = 0.05$ for the third alternative. In the right half of the table, the values in the parentheses are the power of TST-MD (left) and TST-S (right), respectively.}
\label{table:power3}
\begin{tabular*}{\textwidth}{@{\extracolsep{\fill}}cccccccc}
\hline
\multirow{2}{*}{} & \multicolumn{3}{c}{$T_n$}   & \multicolumn{3}{c}{$T_n^{boot}$}  \\ \cline{2-4} \cline{6-8} 
                  & $r=0.01$ & $r=0.05$ & $r=0.1$ & & $r=0.01$ & $r=0.05$ & $r=0.1$ \\ \hline
$K=2$ & 1 & 1 & 1 & & 1 (1, 1) & 1 (1, 1) & 1 (1, 1) \\
$K=3$ & 1 & 1 & 1 & & 1 (1, 1) & 1 (1, 1) & 1 (1, 1) \\
$K=4$ & 1 & 1 & 1 & & 1 (1, 1) & 1 (1, 1) & 1 (1, 1) \\
$K=5$ & 1 & 1 & 1 & & 1 (1, 1) & 1 (1, 1) & 1 (1, 1) \\
$K=6$ & 1 & 1 & 1 & & 1 (1, 1) & 1 (1, 1) & 1 (1, 1) \\
$K=10$ & 1 & 1 & 1 & & 1 (1, 1) & 1 (1, 1) & 1 (1, 1) \\ \hline
\end{tabular*}
\end{table*}

\section{Real Example}\label{Real}
In this section, we apply the proposed method to two real network data. We analyze two Twitter networks \cite{Greene:2013} (the original data used can be found at \url{http://mlg.ucd.ie/networks}). The first network data is the Twitter social network between 419 members of parliament (MPs) from the United Kindom. There are many Twitter social relationships in the real world, such as `follows', `reweeter', and `mentions'. Hence, the different relationships may lead to different networks. For example, for the relationship `follows', the edge between MP $i$ and MP $j$ exists when MP $i$ (or MP $i$) follows MP $j$ (or MP $j$). Other relationships are similar to the relationship `follows'. In addition, all members of parliament can be divided into five groups according to political parties (Conservatives, Labour, Liberal Democrats, SNP, and others.). Here, we delete the 38 members who have no connection with other members in two relationships `follows' and `reweeter', 7 members who do not belong to any political party, as well as 5 SNP members. Finally, we get two networks $\cX_1$ and $\cY_1$ corresponding to relationships `follows' and `reweeter', respectively, which contain 369 nodes. Using the proposed testing procedure, we obtain $T_n=399.0579$ and $T_n^{boot}=374.0122$. Since, $t_{0.975}=1.453$ for Tracy-Widom distribution, we reject $H_0:(g_{x_1},B_{x_1})=(g_{y_1},B_{y_1})$ at the level of 0.05. The first network data is the Twitter social network between 276 politicians from the five parties in the Republic of Ireland (Fine Gael, Labour Party, Fianna F\`{a}il, Sinn F\`{e}in, and Green Party). Similarly, we get two social networks $\cX_2$ and $\cY_2$ corresponding to relationships `follows' and `reweeter', respectively, which contain 276 nodes. Then we calculate the values of test statistics, and obtain $T_n=318.9236$ and $T_n^{boot}=324.2086$. Hence, we reject $H_0:(g_{x_2},B_{x_2})=(g_{y_2},B_{y_2})$ at the level of 0.05.

\begin{figure*}[h]
\centering
\vspace{-0.35cm}
\setlength{\abovecaptionskip}{-2pt}
\subfigtopskip=2pt
\subfigbottomskip=2pt
\subfigcapskip=-5pt
\subfigure[uk-follows]{
\label{pcolor.sub.m2}
\includegraphics[width=0.23\linewidth]{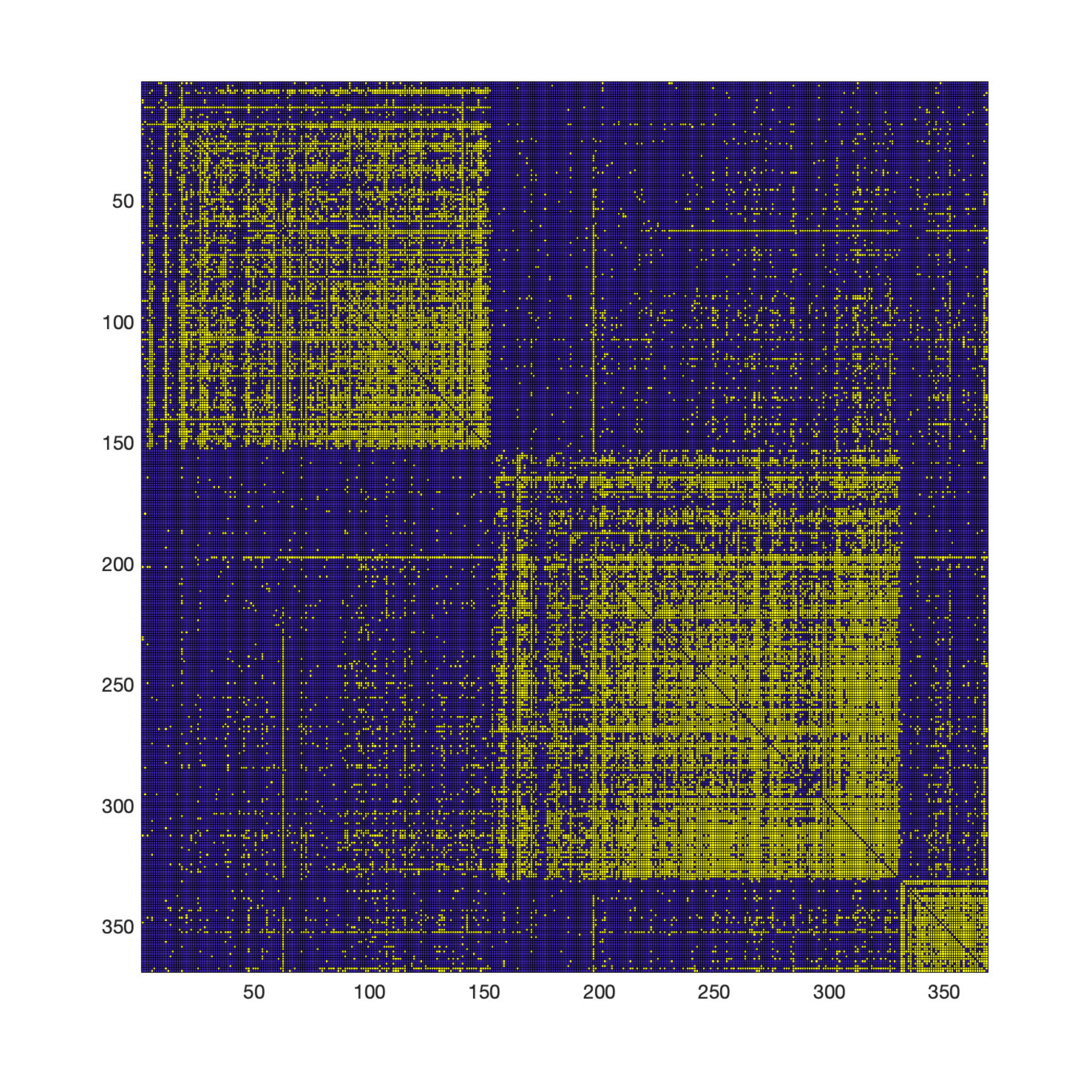}} 
\subfigure[uk-reweets]{
\label{pcolor.sub.m3}
\includegraphics[width=0.23\linewidth]{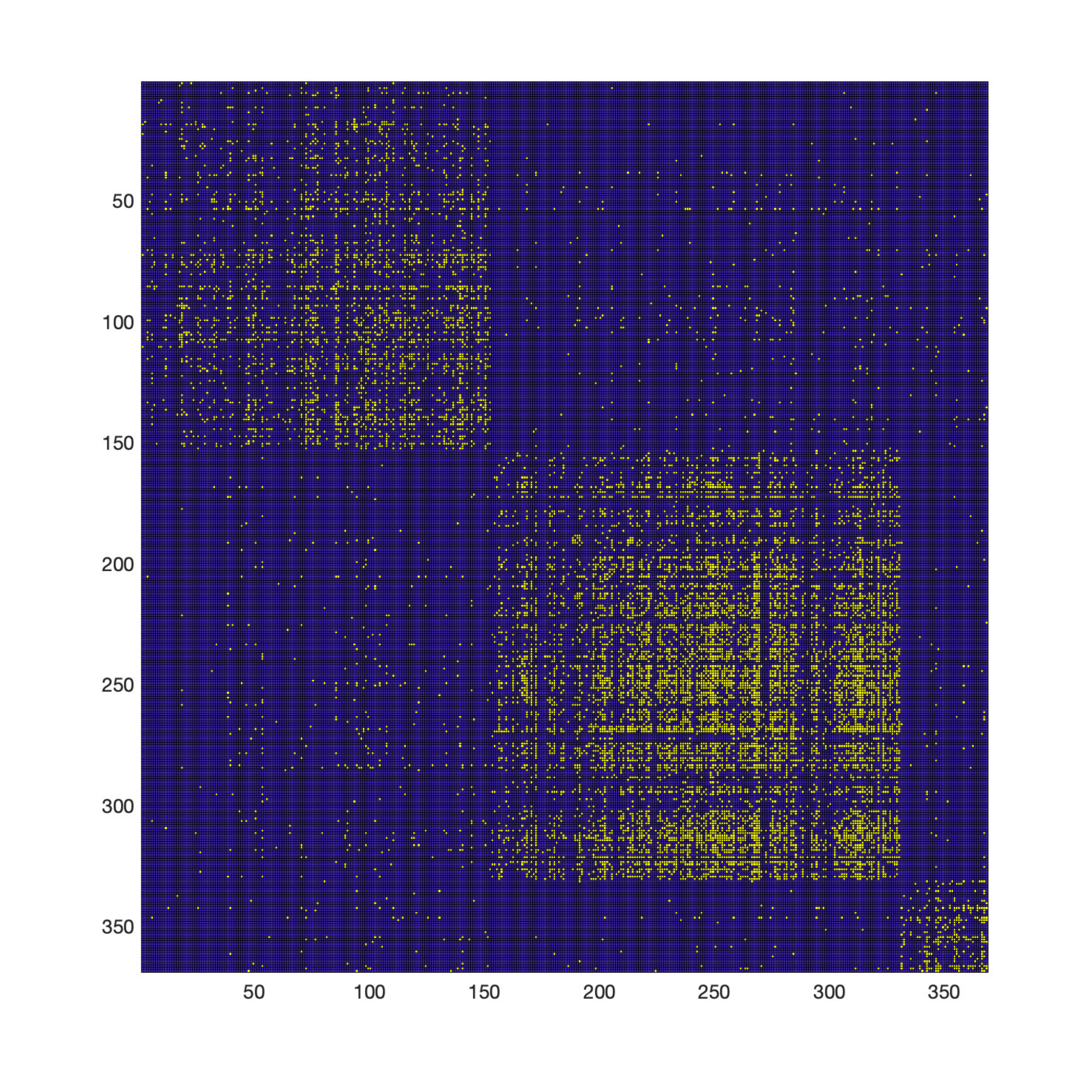}}
\subfigure[ie-follows]{
\label{pcolor.sub.m2}
\includegraphics[width=0.23\linewidth]{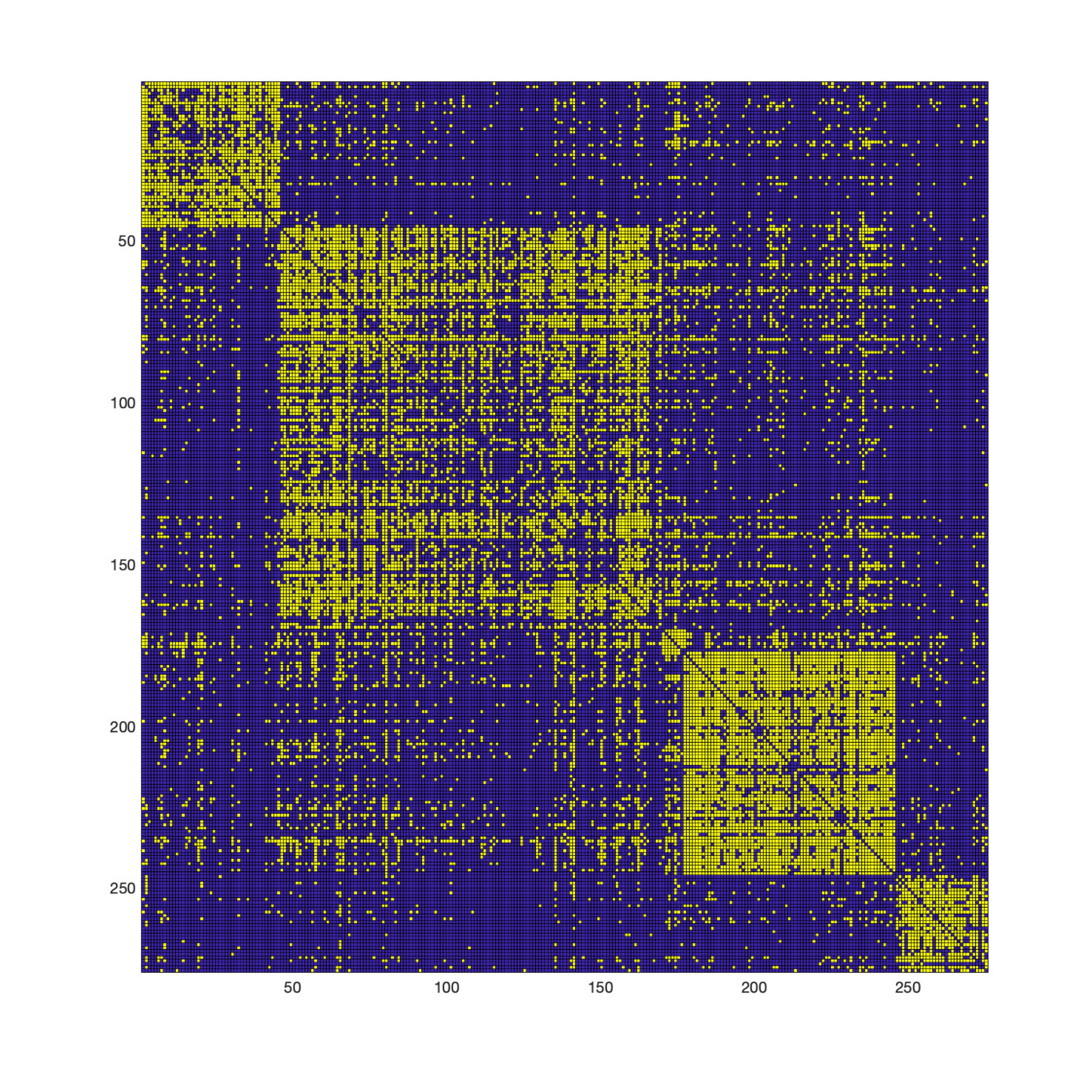}}
\subfigure[ie-reweets]{
\label{pcolor.sub.m3}
\includegraphics[width=0.23\linewidth]{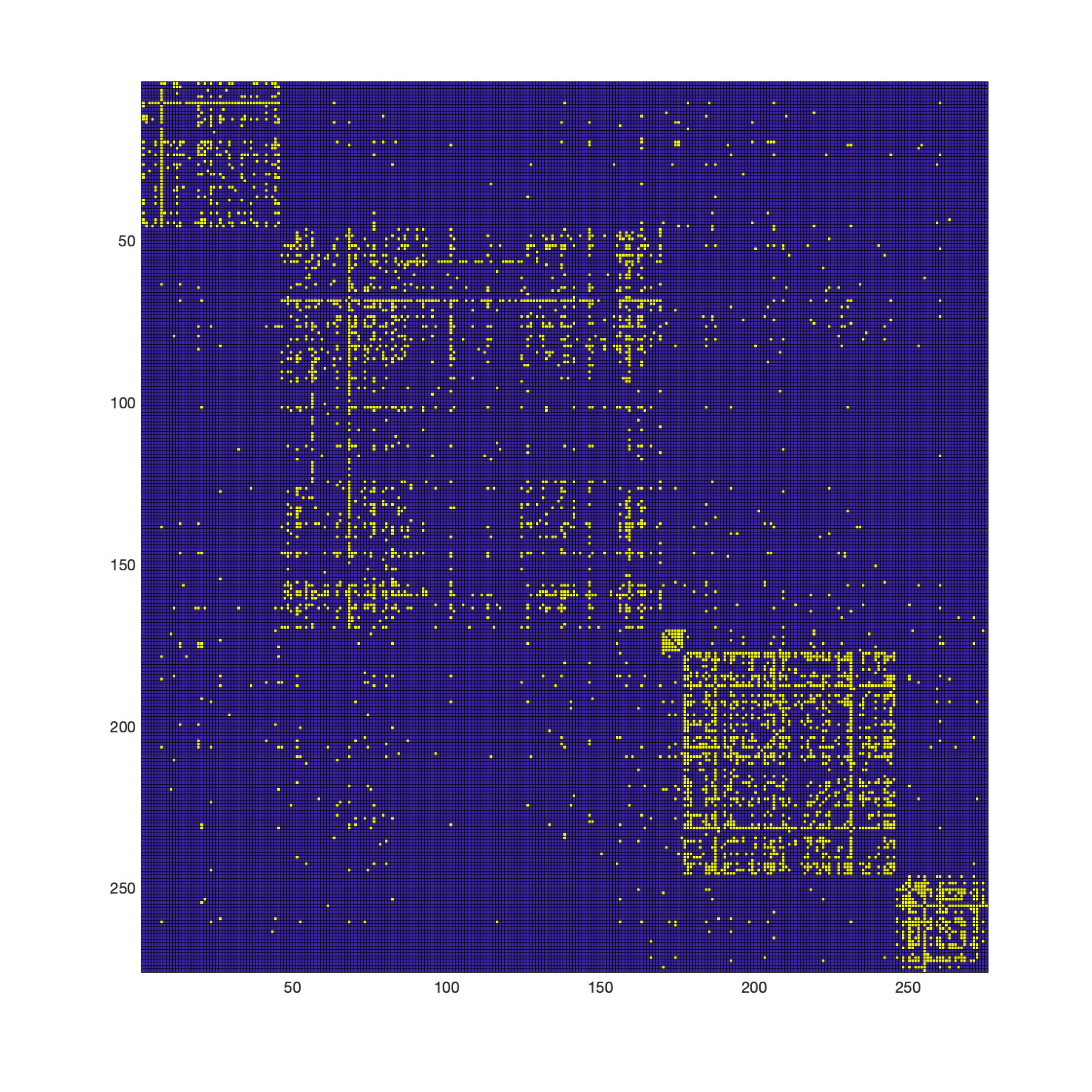}}
\vspace{0.5cm}
\caption{The adjacency matrices corresponding to the two networks of the UK members of the Parliaments Twitter network (two figures on the left) and the two networks of the Irish politicians' Twitter network (two figures on the right), where the rows/columns are sorted with respect to the political party memberships.}
\label{fig:1}
\end{figure*}

Figure \ref{fig:1} plots the visual representation of the four social networks. It is easy to see that in the two samples of the two data, although the network has the same community label, the connection probability is obviously different. The connection probability of the 'reweter' relationship is significantly lower than that of the 'follows' relationship, which implies that $K_{x_1}=K_{y_1}$ and $g_{x_1}=g_{y_1}$ but $B_{x_1}=B_{y_1}$, and $\cX_2$ and $\cY_2$ has similar results. Hence, for the two network data, we should reject the null hypothesis, which is consistent with the numerical results.

\section{Discussion}\label{Discussion}

In this paper, we have proposed a new test statistic for two samples of stochastic block models. Specifically, we first consider combining two samples into one sample to obtain a new observed matrix, then obtain the novel test statistic based on the largest singular value of a residual matrix obtained by subtracting the geometric mean of two estimated block mean effects from the sum of two observed adjacency matrices. Correspondingly, we have demonstrated that its asymptotic null distribution is a Tracy-Widom distribution when $K_x = K_y = o(n^{1/6})$. This test has extended the method proposed by \cite{Lei:2016} of testing one sample to two samples. Empirically, we have demonstrated that the size and power of the test are valid.

It is worth noting that our test method does not work well when the network is spare. However, we may wonder if it can be extended to the sparse case. In the sparse network, the accuracy of the community label estimation may be low. For this case, it is possible to use the community detection algorithm in the sparse network, for instance, regularized spectral clustering \citep{Amini:2013}, profile-pseudo likelihood methods \citep{Wang:2021}. However, in the real network, the degree-corrected stochastic block model and the mixed membership stochastic block model may be more widely used. The proposed method may not work when applied directly to these two models. We need to consider other statistics. In addition, although the proposed can separate the null hypothesis and alternative hypothesis, the four alternative cases are difficult to be distinguished, which is crucial in the real application. Intuitively, we can judge whether $g_x$ is equal to $g_y$. But, it is no theoretical guarantee. Hence, it would be of interest to investigate the possibility of these issues in future work.

\section{Appendix}
First, we give some additional notations and results from random matrix theory for the next proof. Let $(\tilde{\lambda}_j, \tilde{u}_j)$ be the be the eigenvalue-eigenvector pairs of $\tilde{A}^*$ such that $\tilde{\lambda}_1 \geq \tilde{\lambda}_2 \geq \cdots \geq \tilde{\lambda}_n$. For a pair of random sequences $\{a_n\}$ and $\{b_n\}$, we write $a_n = \tilde{O}_p(b_n)$ if for any $\varepsilon>0$ and $D>0$ there exists $n_0=n_0(\varepsilon, D)$ such that
\[
\bbP\{a_n \geq n^{\varepsilon}b_n\}\leq n^{-D}\qquad\mathrm{for all}\ n\geq n_0.
\]
For any matrix $M$, we have $M=\sum_j\sigma_ju_jv_j^\top$ by singular value decomposition, then define $|M|=\sum_j|\sigma_j|u_jv_j^\top$. We will use $c$ and $C$ to denote positive constants independent of $n$, which may vary from line to line.

Now, we give some results from random matrix theory regarding the distributions of the eigenvalues and eigenvectors.

\begin{Lem}\label{Lemma:1}
Under $H_0: (g_x, B_x) = (g_y, B_y)$, for $\tilde{A}^*$ defined in \eqref{eq:As} we have
\[
n^{2/3}[\lambda_1(\tilde{A}^*)-2]\stackrel{d}{\longrightarrow} TW_1\ \text{and}\ n^{2/3}[-\lambda_n(\tilde{A}^*)-2]\stackrel{d}{\longrightarrow} TW_1.
\]
\end{Lem}

\begin{proof}
	Let $G$ be a $n\times n$ symmetric matrix whose upper diagonal entries are independent normal with mean 0 and variance $1/(n-1)$, and diagonal entries are 0. Then, under $H_0$, $\tilde{A}^*$ and $G$ have the same first and second moments. According to Theorem 2.4 of \cite{Erdos:2012}, we have that $n^{2/3}[\lambda_1(\tilde{A}^*)-2]$ and $n^{2/3}[\lambda_1(G)-2]$ have the same limiting distribution. On the other hand, according to \cite{Lee:2014}, $n^{2/3}[-\lambda_n(G)-2]\stackrel{d}{\longrightarrow} TW_1$. Therefore, the results hold.
\end{proof}

\begin{Lem}\label{Lemma:2}
For each deterministic unit vector $u$ and each $1 \leq j \leq n$, for any $\varepsilon>0$ and $D>0$ there exists $n_0 = n_0(\varepsilon, D)$ such that
\[
\bbP\{(u^\top \tilde{u}_j)^2\geq n^{-1+\varepsilon}\}\leq n^{-D}\qquad \text{for all}\ n\geq n_0.
\]
\end{Lem}

Note that the above result is uniform over $j$ and $u$ in the case that $n_0$ does depend on $j$ and $u$. Using above notation, the result can be write as $(u^\top\tilde{u}_j)^2=\tilde{O}_p(n^{-1})$ uniformly over all $\tilde{u}_j (1\leq j\leq n)$ and all deterministic unit vector $u$.

\begin{Lem}[\cite{Lei:2016}]\label{Lemma:3}
	Let $c_n$ be a possibly random number of order $o_p(1)$ and $m(c_n)$ be the number of eigenvalues of $\tilde{A}^*$ larger than $\tilde{\lambda}_1 - c_n$. Then $m(c_n) = O_p(nc_n^{3/2})+\tilde{O}_p(1)$.
\end{Lem}

\subsection{Proof of Theorem \ref{Thm:null}}
According to consistent assumptions of the number of communities and community membership labels, we will focus on the event $\{\widehat{g}_x = g_x, \widehat{g}_y = g_y, \widehat{K}_x = K_x, \widehat{K}_y = K_y\}$.

First, we consider the result for $\tilde{A}$. The other result can be proved by applying the same argument on $-\tilde{A}$. For ease of presentation, we write $g = g_x = g_y$ and $B = B_x = B_y$ when $H_0$ holds.

Let $\tilde{A}'\in\bbR^{n\times n}$ be such that
\begin{equation}
	\tilde{A}'=
	\begin{dcases}
		\dfrac{X_{ij}+Y_{ij} - 2\widehat{P}_{ij}}{\sqrt{(n-1)\cdot 2P_{ij}(1-P_{ij})}}, & i\neq j,\\
		\dfrac{2P_{ii} - 2\widehat{P}_{ii}}{\sqrt{(n-1)\cdot 2P_{ii}(1-P_{ii})}}, & i=j.
	\end{dcases}
\end{equation}
Thus, $\tilde{A}'=\tilde{A}^*+\Delta'$, where $\Delta_{ij}'=(2P_{ij} - 2\widehat{P}_{ij})/\sqrt{(n-1)\cdot 2P_{ij}(1-P_{ij})}$. Under $H_0$, we can see that $P_{ij} = P^x_{ij} = P^y_{ij}$. Hence, $\Delta'$ is a $K\times K$ block-wise constant symmetric matrix, its rank is at most $K$, and the corresponding principal subspace is spanned by $(\theta_1, \cdots, \theta_K)$, where $\theta_k\in\bbR^n$ is the unit norm indicator of the $k$th community in $g$. That is, the $i$th entry of $\theta_k$ is $n_k^{-1/2}$ if $g(i) = k$ and zero otherwise, where $n_k$ is the size of the $k$th community.

On the other hand, by Hoeffding's inequality, we have $\max_{u,v}|\widehat{B}_{uv}^x - B_{uv}^x| = \max_{u,v}|\widehat{B}_{uv}^y - B_{uv}^y| = o_p(K\log n/n)$, which implies that $\max_{i,j}|\widehat{P}_{ij}^x - P_{ij}^x| = \max_{i,j}|\widehat{P}_{ij}^y - P_{ij}^y| = o_p(K\log n/n)$. Since $\dfrac{\widehat{P}_{ij}^x + \widehat{P}_{ij}^y}{2} \geq \sqrt{\widehat{P}_{ij}^x \cdot \widehat{P}_{ij}^y}$, we have
\begin{align*}
	& \max_{ij}|2\widehat{P}_{ij}-2P_{ij}| \leq 2\max_{i,j}|\widehat{P}_{ij} - P_{ij}| \\ &\leq 2\max_{i,j}|\sqrt{\widehat{P}_{ij}^x \cdot \widehat{P}_{ij}^y} - P_{ij}| \leq 2\max_{i,j}|\dfrac{\widehat{P}_{ij}^x + \widehat{P}_{ij}^y}{2} - P_{ij}|,
\end{align*}
which implies $\max_{ij}|2\widehat{P}_{ij}-2P_{ij}| = o_p(K\log n/n)$. Let $\Delta' = \Theta\Gamma\Theta^\top$, where $\Theta = (\theta_1, \cdots, \theta_k)$ and $\Gamma$ is a $K\times K$ symmetric matrix. Then each entry of $\Gamma$ is $o_p(n^{-1/2}\log n)$, and hence $\|\Gamma\|=o_p(Kn^{-1/2}\log n)$.

We will show that
\begin{equation}\label{eq:proof1}
\lambda_1(\tilde{A}') = \lambda_1(\tilde{A}^*) + o_p(n^{-2/3}),
\end{equation}
by establishing a lower and upper bound on $\lambda_1(\tilde{A}')$. Both parts use the eigenvector delocalization result (Lemma \ref{Lemma:2}) as follows. Let $\Theta = (\theta_1, \cdots, \theta_k)$, then, uniformly over $j$ we have
\[
	\|\Theta^\top\tilde{u}_j\|^2_2=\sum_{k=1}^K(\theta_k^\top\tilde{u}_j)^2 = \tilde{O}_p(Kn^{-1}),
\]
and hence
\begin{align*}
	&|\tilde{u}_j^\top\Delta'\tilde{u}_j|\leq|(\Theta^\top\tilde{u}_j)^\top\Gamma(\Theta^\top\tilde{u}_j)|\leq\|\Theta^\top\tilde{u}_j\|_2^2\|\Gamma\|\\ & =\tilde{O}_p(K^2n^{-3/2}\log n),
\end{align*}

where the last equality holds when taking union bound over $K$ terms by choosing $D$ large enough in Lemma \ref{Lemma:2}.

Now, we prove a lower bound on $\lambda_{1}(\tilde{A}')$:
\begin{align*}
	\lambda_1(\tilde{A}') & \geq\ \tilde{u}_1^\top\tilde{A}'\tilde{u}_1\\
	& = \tilde{\lambda}_1 + \tilde{u}_1^\top\Delta'\tilde{u}_1 \\ 
	& \geq \tilde{\lambda}_1 - \tilde{O}_p(K^2n^{-3/2}\log n), \\
	& \geq \tilde{\lambda}_1 - o_p(n^{-2/3}),
\end{align*}
where the last inequality uses the assumed upper bound on the rate at which $K$ grows with $n$.

Similarly, we also prove an upper bound on $\lambda_1(\tilde{A}')$. For any unit vector $u\in\bbR^{n}$, let $(a_1,\cdots,a_n)$ be a unit vector in $\bbR^{n}$ such that $u = \sum_{j=1}^n a_j\tilde{u}_j$.

Let $m$ be the number of $\tilde{\lambda}_j$’s in the interval $(\tilde{\lambda}_1 - 2\|\Delta'\|, \tilde{\lambda}_1]$, and $u_1 = \sum_{j=1}^m a_j\tilde{u}_j, u_2 = \sum_{j=m+1}^n a_j\tilde{u}_j$. Then
\begin{align*}
	u^\top\tilde{A}'u & = u^\top\tilde{A}^*u + u^\top\Delta'u \notag \\
	& \leq \tilde{\lambda}_1\sum_{j=1}^ma_j^2 + (\tilde{\lambda}_1 - 2\|\Delta'\|)\sum_{j=m+1}^na_j^2 + \\
	& \qquad 2u_1^\top|\Delta'|u_1 + u_2^\top|\Delta'|u_2 \notag \\
	& \leq \tilde{\lambda}_1\sum_{j=1}^ma_j^2 + (\tilde{\lambda}_1 - 2\|\Delta'\|)\sum_{j=m+1}^na_j^2 + \\
	& \qquad 2m \sum_{j=1}^ma_j^2\tilde{u}_j^\top|\Delta'|\tilde{u}_j + u_2^\top|\Delta'|u_2 \notag \\
	& \leq \tilde{\lambda}_1\sum_{j=1}^ma_j^2 + (\tilde{\lambda}_1 - 2\|\Delta'\|)\sum_{j=m+1}^na_j^2 + \\
	& \qquad 2m\tilde{O}_p(K^2n^{-2/3}\log n)\sum_{j=1}^ma_j^2 + 2\|\Delta'\|\sum_{j=m+1}^na_j^2 \notag \\
	& \leq \tilde{\lambda}_1 + m\tilde{O}_p(K^2n^{-2/3}\log n) \notag \\
	& \leq \tilde{\lambda}_1 + (O(n\|\Delta'\|^{3/2})+\tilde{O}_p(1))\tilde{O}_p(K^2n^{-2/3}\log n) \notag \\
	& = \tilde{\lambda}_1 + \tilde{O}_p(K^{7/2}n^{-5/4}(\log n)^{5/2}),
\end{align*}
where the second last line uses Lemma \ref{Lemma:3} together with $\|\Delta'\| = o_p(Kn^{-1/2}\log n)$. Thus, \eqref{eq:proof1} holds when $K=O(n^{1/6-\tau})$ for some small positive $\tau$.

Next, we show that $\lambda_1(\tilde{A}) = \lambda_1(\tilde{A}') + o_p(n^{-2/3})$. Let $\tilde{A}'' = \tilde{A}' - \diag(\tilde{A}')$. Consider the partitioned matrix of $\tilde{A}$:
\[
\tilde{A} = (\tilde{A}_{(k,l)})_{k,l=1}^K,
\]
where $\tilde{A}_{(k,l)}$ is the submatrix corresponding to the rows in community $k$ and columns in community $l$. Similar notations can be defined for $\tilde{A}''$.

Note that 
\[
\tilde{A}_{(k,l)} = \tilde{A}''_{(k,l)}\sqrt{\dfrac{B_{kl}(1-B_{kl})}{\widehat{B}_{kl}(1-\widehat{B}_{kl})}} = \tilde{A}''(1+o_p(K\log n/n)).
\]
Therefore,
\begin{align*}
	\|\tilde{A} - \tilde{A}''\| & = K\max_{k,l}\|\tilde{A}_{(k,l)} - \tilde{A}''_{(k,l)}\| \\
	& \leq o_p(K\log n/n)K\sum_{k,l}\|\tilde{A}''_{(k,l)}\| \\
	& \leq o_p(K^2\log n/n)\|\tilde{A}''\| \\
	& \leq o_p(K^2\log n/n)[\|\tilde{A}'\|+\|\diag(\tilde{A}')\|] \\
	& \leq o_p(K^2\log n/n)[O_p(1)+O_p(Kn^{-3/2}\log n)] \\
	& = o_p(K^2\log n/n) = o_p(n^{-2/3}).
\end{align*}

Then 
\[
\|\tilde{A} - \tilde{A}'\|\leq\|\tilde{A} - \tilde{A}''\| + \|\diag(\tilde{A}'')\| = o_p(n^{-2/3}),
\]
combining with \eqref{eq:proof1}, we have
\[
\lambda_1(\tilde{A}') = \lambda_1(\tilde{A}^*) + o_p(n^{-2/3}).
\]
By applying Lemma \ref{Lemma:1}, we have
\[
n^{-2/3}[\lambda_1(\tilde{A})-2]\stackrel{d}{\longrightarrow} TW_1.
\]
\hfill$\square$

\subsection{Proof of Theorem \ref{Thm:power}}
For $1\leq l_x \leq K_x, 1\leq k_x \leq K_x$, denote $\cN_{l_x}^x = \{i:g_x(i)=l_x\}, \widehat{\cN}_{k_x}^x = \{i:\widehat{g}_x(i)=k_x\}$, and $\widehat{\cN}_{k_x,l_x}=\{i:\widehat{g}_x(i)=k_x,g_x(i)=l_x\}$. Note, for each $1\leq l_x \leq K_x$, $\cN_{l_x}^x=\sum_{k_x=1}^{K_x}\widehat{\cN}_{k_x,l_x}$, where symbol $\sum$ represents the union of pairwise disjoint sets. Thus, for each $1\leq l_x \leq K_x$ there exist a $1\leq k_x(l_x) \leq K_x$, which depend on $l_x$, such that $|\widehat{\cN}_{k_x(l_x), l_x}| \geq |\widehat{\cN}_{l_x}|/K_x \geq c_xn/K_x^2$. Similarly, the same results hold when the script $x$ is replaced by $y$, and we omit the details. Since, $K_x = K_y$, there exist $l_1, l_2, l_3$ and $l_4$ such that $k_x(l_1) = k_x(l_2) = k_y(l_3) = k_y(l_4) = k$. By Assumption \ref{Ass:1}, there exist $l_5$ and $l_6$ such that $B_{l_1, l_5}^x\neq B_{l_2, l_5}^x$ and $B_{l_3, l_6}^y\neq B_{l_4, l_6}^y$, and we have $|\widehat{\cN}_{k_x(l_5),l_5}^x| \geq c_x n/K_x^2$ and $|\widehat{\cN}_{k_y(l_6),l_6}^y| \geq c_y n/K_y^2$.

Let $\mathcal{A}$ be the submatrix of $\tilde{A}$ consisting the rows in $\widehat{\cN}_{k,l_1}^x\cup\widehat{\cN}_{k,l_2}^x\cup\widehat{\cN}_{k,l_3}^y\cup\widehat{\cN}_{k,l_4}^y$, and the columns in $\widehat{\cN}_{k_x(l_5),l_5}^x\cup \widehat{\cN}_{k_y(l_6),l_6}^y$. For simplicity, denote $Z = \cX + \cY$. Hence, $\mathcal{Z}, \mathcal{P}$ and $\widehat{\mathcal{P}}$ be defined in the same way as $\mathcal{A}$. 

Note, when $k \neq k_x(l_5) \neq k_y(l_6)$, or $k = k_x(l_5) = k_y(l_6)$ but $l_5\notin\{l_1, l_2\}, l_6\notin\{l_3, l_4\}$, the submatrix $\mathcal{A}$ contains only off-diagonal entries of $A$. Therefore, $\widehat{\mathcal{P}}$ is a constant matrix. It holds that
\begin{align}\label{eq:proof2}
	\|\tilde{A}\| & \geq \|\mathcal{A}\| \geq [2(n-1)]^{-1/2}\|\mathcal{A} - 2\widehat{\mathcal{P}}\| \notag \\
	& \geq (2n)^{-1/2} (\|2\mathcal{P} - \widehat{2\mathcal{P}}\| - \|\mathcal{A} - 2\mathcal{P}\|) \notag \\
	& \geq (2n)^{-1/2}[2\|\mathcal{P} - \widehat{\mathcal{P}}\| - O_p(\sqrt{2n})] \notag \\
	& \geq (2n)^{-1/2}[(\delta_x\wedge\delta_y)(c_x\wedge c_y)nK^{-2}/2 - O_p(\sqrt{2n})].
\end{align}

To obtain the last inequality, first, note that $\mathcal{P}_x$ has two distinct blocks each with at least $c_xn/K^2$ rows and at least $c_xn/K^2$ columns, $\mathcal{P}_y$ also has two distinct blocks each with at least $c_yn/K^2$ rows and at least $c_yn/K^2$ columns. Thus, $\mathcal{P} = \mathcal{P}_x * \mathcal{P}_y$ has two distinct blocks each with at least $(c_x\wedge c_y)n/K^2$ rows and at least $(c_x\wedge c_y)n/K^2$ columns, where $*$ indicate Hadamard product of two matrices. Each of these two blocks has constant entries and at least one of them has an absolute entry value of at least $(\delta_x\wedge\delta_y)/2$. Thus, $2\|\mathcal{P} - \widehat{\mathcal{P}}\| \geq (\delta_x\wedge\delta_y)(c_x\wedge c_y)nK^{-2}/2$.

Finally, when $k = k_x(l_5) = k_y(l_6)$, and $l_5\in\{l_1, l_2\}, l_6\in\{l_3, l_4\}$, the submatrix $\mathcal{A}$ contains diagonal entries of $\mathcal{A}$. The corresponding entries of $\widehat{\mathcal{P}}$ are 0. These zero entries lead an additional $O(1)$ term in $2\|\mathcal{P} - \widehat{\mathcal{P}}\|$ and \eqref{eq:proof2} still holds. \hfill $\square$

\bibliographystyle{agu04}
\bibliography{ref}

\begin{thebibliography}{36}
\providecommand{\natexlab}[1]{#1}
\expandafter\ifx\csname urlstyle\endcsname\relax
  \providecommand{\doi}[1]{doi:\discretionary{}{}{}#1}\else
  \providecommand{\doi}{doi:\discretionary{}{}{}\begingroup
  \urlstyle{rm}\Url}\fi

\bibitem[{\textit{Airoldi et~al.}(2008)\textit{Airoldi, Blei, Fienberg, and
  Xing}}]{Airoldi:2008}
Airoldi, E.~M., D.~M. Blei, S.~E. Fienberg, and E.~P. Xing (2008), Mixed
  membership stochastic blockmodels, \textit{Journal of machine learning
  research}, \textit{9}, 1981 -- 2014.

\bibitem[{\textit{Amini et~al.}(2013)\textit{Amini, Chen, Bickel, and
  Levina}}]{Amini:2013}
Amini, A.~A., A.~Chen, P.~J. Bickel, and E.~Levina (2013), Pseudo-likelihood
  methods for community detection in large sparse networks,, \textit{The Annals
  of Statistics}, \textit{41}(4), 2097 -- 2122, \doi{10.1214/13-AOS1138}.

\bibitem[{\textit{Barnett and Onnela}(2016)}]{Barnett:2016}
Barnett, I., and J.-P. Onnela (2016), Change point detection in correlation
  networks, \textit{Scientic reports}, \textit{6}, 18,893,
  \doi{10.1038/srep18893}.

\bibitem[{\textit{Bickel and Chen}(2009)}]{Bickel:2009}
Bickel, P.~J., and A.~Chen (2009), A nonparametric view of network models and
  newman-girvan and other modularities, \textit{Proceedings of the National
  Academy of Sciences}, \textit{106}(50), 21,068 -- 21,073,
  \doi{10.1073/pnas.0907096106}.

\bibitem[{\textit{Bickel and Sarkar}(2015)}]{Bickel:2015}
Bickel, P.~J., and P.~Sarkar (2015), Hypothesis testing for automated community
  detection in networks, \textit{Journal of the Royal Statistical Society.
  Series B (Statistical Methodology)}, \textit{78}(1), 253 -- 273,
  \doi{10.1111/rssb.12117}.

\bibitem[{\textit{Chen et~al.}(2021{\natexlab{a}})\textit{Chen, Josephs, Lin,
  Zhou, and Kolaczyk}}]{Chen:2021s}
Chen, L., N.~Josephs, L.~Lin, J.~Zhou, and E.~D. Kolaczyk (2021{\natexlab{a}}),
  A spectral-based framework for hypothesis testing in populations of networks,
  \textit{Statistica Sinica}, \textit{online}, \doi{10.5705/ss.202021.0306}.

\bibitem[{\textit{Chen et~al.}(2021{\natexlab{b}})\textit{Chen, Zhou, and
  Lin}}]{Chen:2021c}
Chen, L., J.~Zhou, and L.~Lin (2021{\natexlab{b}}), Hypothesis testing for
  populations of networks, \textit{Communications in Statistics - Theory and
  Methods}, \textit{online}, \doi{10.1080/03610926.2021.1977961}.

\bibitem[{\textit{Choi et~al.}(2012)\textit{Choi, Wolfe, and
  Airoldi}}]{Choi:2012}
Choi, D.~S., P.~J. Wolfe, and E.~M. Airoldi (2012), Stochastic blockmodels with
  a growing number of classes, \textit{Biometrika}, \textit{99}(2), 273--284,
  \doi{10.1093/biomet/asr053}.

\bibitem[{\textit{Daudin et~al.}(2008)\textit{Daudin, Picard, and
  Robin}}]{Daudin:2008}
Daudin, J.-J., F.~Picard, and S.~Robin (2008), A mixture model for random
  graphs, \textit{Statistics and Computing}, \textit{18}, 173 -- 183,
  \doi{10.1007/s11222-007-9046-7}.

\bibitem[{\textit{Erd{\H{o}}s et~al.}(2012)\textit{Erd{\H{o}}s, Knowles, Yau,
  and Yin}}]{Erdos:2012}
Erd{\H{o}}s, L., A.~Knowles, H.-T. Yau, and J.~Yin (2012), Spectral statistics
  of {E}rd{\H{o}}s-{R}\'enyi graphs {II}: Eigenvalue spacing and the extreme
  eigenvalues, \textit{Communications in Mathematical Physics volume},
  \textit{314}, 587 -- 640, \doi{10.1007/s00220-012-1527-7}.

\bibitem[{\textit{Ghoshdastidar and von Luxburg}(2018)}]{Ghoshdastidar:2018}
Ghoshdastidar, D., and U.~von Luxburg (2018), Practical methods for graph
  two-sample testing, in \textit{Advances in Neural Information Processing
  Systems}, vol.~31, edited by S.~Bengio, H.~Wallach, H.~Larochelle,
  K.~Grauman, N.~Cesa-Bianchi, and R.~Garnett, pp. 3019--3028, Curran
  Associates, Inc.

\bibitem[{\textit{Ghoshdastidar et~al.}(2020)\textit{Ghoshdastidar, Gutzeit,
  Carpentier, and von Luxburg}}]{Ghoshdastidar:2020}
Ghoshdastidar, D., M.~Gutzeit, A.~Carpentier, and U.~von Luxburg (2020),
  Two-sample hypothesis testing for inhomogeneous random graphs, \textit{The
  Annals of Statistics}, \textit{48}(4), 2208--2229, \doi{10.1214/19-AOS1884}.

\bibitem[{\textit{Greene and Cunningham}(2013)}]{Greene:2013}
Greene, D., and P.~Cunningham (2013), Producing a unified graph representation
  from multiple social network views, in \textit{Proceedings of the 5th Annual
  ACM Web Science Conference}, vol. WebSci '13, p. 118–121, Association for
  Computing Machinery, New York, USA, \doi{10.1145/2464464.2464471}.

\bibitem[{\textit{Holland et~al.}(1983)\textit{Holland, Laskey, and
  Leinhardt}}]{Holland:1983}
Holland, P.~W., K.~B. Laskey, and S.~Leinhardt (1983), Stochastic blockmodels:
  First steps, \textit{Social Networks}, \textit{5}(2), 109 -- 137,
  \doi{10.1016/0378-8733(83)90021-7}.

\bibitem[{\textit{Hu et~al.}(2020)\textit{Hu, Qin, Yan, , and Zhao}}]{Hu:2020}
Hu, J., H.~Qin, T.~Yan, , and Y.~Zhao (2020), Corrected bayesian information
  criterion for stochastic block models, \textit{Journal of the American
  Statistical Association}, \textit{115}(532), 1771 -- 1783,
  \doi{10.1080/01621459.2019.1637744}.

\bibitem[{\textit{Hu et~al.}(2021)\textit{Hu, Zhang, Qin, Yan, and
  Zhu}}]{Hu:2021}
Hu, J., J.~Zhang, H.~Qin, T.~Yan, and J.~Zhu (2021), Using maximum entry-wise
  deviation to test the goodness of fit for stochastic block models,
  \textit{Journal of The American Statistical Association}, \textit{116}(535),
  1373 -- 1382, \doi{10.1080/01621459.2020.1722676}.

\bibitem[{\textit{Jin}(2015)}]{Jin:2015}
Jin, J. (2015), Fast community detection by score, \textit{The Annals of
  Statistics}, \textit{43}(1), 57 -- 89, \doi{10.1214/14-AOS1265}.

\bibitem[{\textit{Karrer and Newman}(2011)}]{Karrer:2011}
Karrer, B., and M.~E.~J. Newman (2011), Stochastic blockmodels and community
  structure in networks, \textit{Physical Review. E}, \textit{83}(1), 016,107,
  \doi{10.1103/PhysRevE.83.016107}.

\bibitem[{\textit{Lee and Yin}(2014)}]{Lee:2014}
Lee, J.~O., and J.~Yin (2014), A necessary and sufficient condition for edge
  universality of wigner matrices, \textit{Duke Mathematical Journal},
  \textit{163}(1), 117 -- 173, \doi{10.1215/00127094-2414767}.

\bibitem[{\textit{Lei}(2016)}]{Lei:2016}
Lei, J. (2016), A goodness-of-fit test for stochastic block models, \textit{The
  Annals of Statistics}, \textit{44}(1), 401 -- 424, \doi{10.1214/15-AOS1370}.

\bibitem[{\textit{Lei and Rinaldo}(2015)}]{Lei:2015}
Lei, J., and A.~Rinaldo (2015), Consistency of spectral clustering in
  stochastic block models., \textit{The Annals of Statistics}, \textit{43}(1),
  215 -- 237, \doi{10.1214/14-AOS1274}.

\bibitem[{\textit{Newman}(2006)}]{Newman:2006}
Newman, M. E.~J. (2006), Modularity and community structure in networks,
  \textit{Proceedings of the National Academy of Sciences of the United States
  of America}, \textit{103}(23), 8577 -- 8582, \doi{10.1073/pnas.0601602103}.

\bibitem[{\textit{Newman and Girvan}(2004)}]{Newman:2004}
Newman, M. E.~J., and M.~Girvan (2004), Finding and evaluating community
  structure in networks, \textit{Physical review. E}, \textit{69}(3), 026,113,
  \doi{10.1103/PhysRevE.69.026113}.

\bibitem[{\textit{Nowicki and Snijders}(2001)}]{Nowicki:2001}
Nowicki, K., and T.~Snijders (2001), Estimation and prediction for stochastic
  block structures, \textit{Journal of The American Statistical Association},
  \textit{96}(11), 1077 -- 1087, \doi{10.1198/016214501753208735}.

\bibitem[{\textit{Rohe et~al.}(2011)\textit{Rohe, Chatterjee, and
  Yu}}]{Rohe:2011}
Rohe, K., S.~Chatterjee, and B.~Yu (2011), Spectral clustering and the
  high-dimensional stochastic blockmodel, \textit{The Annals of Statistics},
  \textit{39}(4), 1878 -- 1915, \doi{10.1214/11-AOS887}.

\bibitem[{\textit{Sald{\~n}a et~al.}(2017)\textit{Sald{\~n}a, Yu, and
  Feng}}]{Saldna:2017}
Sald{\~n}a, D.~F., Y.~Yu, and Y.~Feng (2017), How many communities are there?,
  \textit{Journal of Computational and Graphical Statistics}, \textit{26}(1),
  171 -- 181, \doi{10.1080/10618600.2015.1096790}.

\bibitem[{\textit{Sarkar and Bickel}(2015)}]{Sarkar:2015}
Sarkar, P., and P.~J. Bickel (2015), Role of normalization in spectral
  clustering for stochastic blockmodels, \textit{The Annals of Statistics},
  \textit{43}(3), 962 -- 990, \doi{10.1214/14-AOS1285}.

\bibitem[{\textit{Snijders and Nowicki}(1997)}]{Snijders:1997}
Snijders, T., and K.~Nowicki (1997), Estimation and prediction for stochastic
  block-structures for graphs with latent block structure, \textit{Journal of
  Classification}, \textit{14}, 75 -- 100, \doi{10.1007/s003579900004}.

\bibitem[{\textit{Steinhaeuser and Chawla}(2010)}]{Steinhaeuser:2010}
Steinhaeuser, K., and N.~V. Chawla (2010), Identifying and evaluating community
  structure in complex networks, \textit{Pattern Recognition Letters},
  \textit{31}(5), 413 -- 421, \doi{10.1016/j.patrec.2009.11.001}.

\bibitem[{\textit{Tang et~al.}(2017)\textit{Tang, Athreya, Sussman, Lyzinski,
  and Priebe}}]{Tang:2017}
Tang, M., A.~Athreya, D.~L. Sussman, V.~Lyzinski, and C.~E. Priebe (2017), A
  nonparametric two-sample hypothesis testing problem for random graphs,
  \textit{Bernoulli}, \textit{23}(3), 1599 -- 1630, \doi{10.3150/15-BEJ789}.

\bibitem[{\textit{Wang et~al.}(2021)\textit{Wang, Zhang, Liu, Zhu, and
  Guo}}]{Wang:2021}
Wang, J., J.~Zhang, B.~Liu, J.~Zhu, and J.~Guo (2021), Fast network community
  detection with profile-pseudo likelihood methods, \textit{Journal of the
  American Statistical Association}, \textit{online},
  \doi{10.1080/01621459.2021.1996378}.

\bibitem[{\textit{Wang and Bickel}(2017)}]{Wang:2017}
Wang, Y.~R., and P.~J. Bickel (2017), Likelihood-based model selection for
  stochastic block models, \textit{The Annals of Statistics}, \textit{45}(2),
  500 -- 528, \doi{10.1214/16-AOS1457}.

\bibitem[{\textit{Wu et~al.}(2022)\textit{Wu, Kong, and Xu}}]{Wu:2021}
Wu, F., X.~Kong, and C.~Xu (2022), Test on stochastic block model: Local
  smoothing and extreme value theory, \textit{Journal of Systems Science and
  Complexity}, \textit{35}(5), 1535–1556, \doi{10.1007/s11424-021-0154-9}.

\bibitem[{\textit{Zhang and Zhou}(2016)}]{Zhang:2016}
Zhang, A.~Y., and H.~H. Zhou (2016), Minimax rates of community detection in
  stochastic block models, \textit{The Annals of Statistics}, \textit{44}(5),
  2252 -- 2280, \doi{10.1214/15-AOS1428}.

\bibitem[{\textit{Zhao et~al.}(2011)\textit{Zhao, Levina, and Zhu}}]{Zhao:2011}
Zhao, Y., E.~Levina, and J.~Zhu (2011), Community extraction for social
  networks, \textit{Proceedings of the National Academy of Sciences of the
  United States of America}, \textit{108}(18), 7321 -- 7326,
  \doi{10.1073/pnas.1006642108}.

\bibitem[{\textit{Zhao et~al.}(2012)\textit{Zhao, Levina, and Zhu}}]{Zhao:2012}
Zhao, Y., E.~Levina, and J.~Zhu (2012), Consistency of community detection in
  networks under degree-corrected stochastic block models, \textit{The Annals
  of Statistics}, \textit{40}(4), 2266--2292, \doi{10.1214/12-AOS1036}.

\end{thebibliography}

\end{document}